\documentclass[fontsize=9pt]{llncs}
  \usepackage[noend]{algpseudocode}
  \usepackage{algorithmicx}
  \usepackage[ruled]{algorithm}
  \usepackage{algpseudocode}
  \usepackage{amsmath,amsfonts}
  \usepackage{tikz}
  \usetikzlibrary{trees}
  \usetikzlibrary{decorations.pathmorphing}
  \usetikzlibrary{decorations.markings}
  \usetikzlibrary{decorations.pathmorphing,shapes}
  \usetikzlibrary{calc,decorations.pathmorphing,shapes}
  \usepackage[T1]{fontenc}
  \usepackage[utf8]{inputenc}
  \usepackage{comment}
  \usepackage{forest}
  \usepackage{xspace}
  \usepackage{enumerate}
  \usepackage{listings}
  \usepackage{multirow}
  \usepackage{todonotes}  
  \usepackage{thmtools}
  \usepackage{thm-restate}
  \usepackage[capitalise]{cleveref}
  \usepackage{graphics,adjustbox}
  \usepackage{tabularx}
  \usepackage{amssymb}
  \usepackage{epsfig}
  \usepackage[whileod]{chl-alg}
  \usepackage[caption=true,font=footnotesize]{subfig}
  \usepackage{verbatim}
  \usepackage{url}
  \usetikzlibrary{trees, arrows, shapes, snakes}
\tikzset{
  treenode/.style = {align=center, inner sep=2pt, text centered,
    font=\sffamily},
  arn_r/.style = {treenode, circle, black, font=\sffamily\bfseries, draw=black,
    text width=1.5em},
    arn_t/.style = {treenode, circle, black, thick, double, font=\sffamily\bfseries, draw=black,
    text width=1.5em},
  every edge/.append style={anchor=south,auto=falseanchor=south,auto=false,font=3.5 em},
}
\usepackage{enumitem,kantlipsum}

\newtheorem{lem}[theorem]{Lemma}
\renewenvironment{lemma}{\begin{lem}}{\end{lem}}
\crefname{lem}{Lemma}{Lemmas}

\def\dd{\mathinner{.\,.}}
\newcommand{\cO}{\mathcal{O}}

\newcommand{\SA}{\textsf{SA}}
\newcommand{\LCP}{\textsf{LCP}}
\newcommand{\PLCP}{\textsf{PLCP}}
\newcommand{\POS}{\textsf{P}}
\newcommand{\LCE}{\textsf{LCE}}

\newcommand{\lcp}{\textsf{lcp}}

\newcommand{\iSA}{\textsf{iSA}}

 \newcommand{\defproblem}[3]{
  \vspace{2mm}
\noindent\fbox{
  \begin{minipage}{0.96\textwidth}
  #1\\
  {\bf{Input:}} #2  \\
  {\bf{Output:}} #3
  \end{minipage}
  }
  \vspace{2mm}
}

\begin{document}
\title{Longest Common Prefixes with $k$-Errors and Applications}

\author{
Lorraine A.K. Ayad
\and
Panagiotis Charalampopoulos
\and
Costas S. Iliopoulos
\and
Solon P. Pissis
}

\institute{
    Department of Informatics, King's College London, London, UK\\
    \email{[lorraine.ayad,panagiotis.charalampopoulos,\\costas.iliopoulos,solon.pissis]@kcl.ac.uk}}

\date{}

\maketitle 
   
\begin{abstract}
Although real-world text datasets, such as DNA sequences, are far from being uniformly random, average-case string searching algorithms perform significantly better than worst-case ones in most applications of interest. In this paper, we study the problem of computing the longest prefix of each suffix of a given string of length $n$ over a constant-sized alphabet that occurs elsewhere in the string with $k$-errors. This problem has already been studied under the Hamming distance model. 
Our first result is an improvement upon the state-of-the-art average-case time complexity for {\em non-constant $k$} and using only {\em linear space} under the Hamming distance model. Notably, we show that our technique can be extended to the edit distance model with the same time and space complexities. Specifically, our algorithms run in $\cO(n \log^k n \log \log n)$ time on average using $\cO(n)$ space. We show that our technique is applicable to several algorithmic problems in computational biology and elsewhere. 
\end{abstract}

\section{Introduction}\label{sec:intro}

The longest common prefix (LCP) array is a commonly used data structure alongside the suffix array (SA). The LCP array stores the length of the longest common prefix between two adjacent suffixes of a given string as they are stored (in lexicographical order) in the SA~\cite{SA}. A typical use combining the SA and the LCP array is to simulate the suffix tree functionality using less space~\cite{eSA}.

However, there are many practical scenarios where the LCP array may be applied without making use of the SA. The LCP array provides us with essential information regarding {\em repetitiveness} in a given string and is therefore a useful data structure for analysing textual data in areas such as molecular biology, musicology, or natural language processing (see~\cite{Manzini:2015:LCP:2952649.2952678} for some applications).

It is also quite common to account for potential alterations within textual data (sequences). For example, they can be the result of DNA replication or sequencing errors in DNA sequences. In this context, it is natural to define the longest common prefix with $k$-errors. Given a string $x[0 \dd n-1]$, the longest common prefix with $k$-errors for every suffix $x[i\dd n-1]$ is the length of the longest common prefix of $x[i\dd n-1]$ and any $x[j\dd n-1]$, where $j \neq i$, with applying up to $k$ substitution operations~\cite{Manzini:2015:LCP:2952649.2952678}. Some applications are given below.

\paragraph{Interspersed Repeats.} Repeated sequences are a common feature of genomes. One type in particular, interspersed repeats, are known to occur in all eukaryotic genomes. These repeats have no repetitive pattern and appear irregularly within DNA sequences \cite{mreps}. Single nucleotide polymorphisms result in the existence of interspersed repeats that are not identical~\cite{bioBk}. Identifying these repeats has been linked to genome folding locations and phylogenetic analysis~\cite{repeats}.

\paragraph{Genome Mappability Data Structure.} In~\cite{Sofsem2018} the authors showed that using the longest common prefixes with $k$-errors they can construct, in $\cO(n)$ worst-case time, an $\cO(n)$-sized data structure answering the following type of queries in $\cO(1)$ time per query: find the smallest $m$ such that at least $\mu$ of the substrings of $x$ of length $m$ do not occur more than once in $x$ with at most $k$ errors. This is a data structure version of the genome mappability problem~\cite{biopaper,Manzini:2015:LCP:2952649.2952678,1map}.

\paragraph{Longest Common Substring with $k$-Errors.} The longest common substring with $k$-errors problem has received much attention recently, in particular due to its applications in computational biology~\cite{DBLP:journals/jcb/UlitskyBTC06,kmacs,ALFRED}. We are asked to find the longest substrings of two strings that are at distance at most $k$. The notion of longest common prefix with $k$-errors is thus {\em closely related} to the notion of longest common substring with $k$-errors. We refer the interested reader to~\cite{Abboud:2015:MAP:2722129.2722146,DBLP:journals/ipl/FlouriGKU15,DBLP:journals/ipl/Grabowski15,tania_arxiv,DBLP:journals/jcb/ThankachanAA16}.

\paragraph{All-Pairs Suffix/Prefix Overlaps with $k$-Errors.} Finding approximate overlaps is the first stage of most genome assembly methods. Given a set of strings and an error-rate $\epsilon$, the goal is to find, for all pairs of strings, their suffix/prefix matches (overlaps) that are within distance $k=\lceil \epsilon\ell \rceil$, where $\ell$ is the length of the overlap~\cite{DBLP:journals/jcb/RasmussenSM06,DBLP:journals/iandc/ValimakiLM12,DBLP:conf/spire/KucherovT14}. By concatenating the strings to form one single string $x$ and then computing longest common prefixes with $k$-errors for $x$ only against the prefixes of the strings we have all the information we need to solve this problem.

\paragraph{Our Model.}
We assume the standard word-RAM model with word size $w = \Omega(\log n)$.
Although real-world text datasets are far from being uniformly random, average-case string searching algorithms perform significantly better than worst-case ones in most applications of interest. We are thus interested in the average-case behaviour of our algorithms.
When we state {\em average-case} time complexities for our algorithms, we assume that the input is a string $x$ of length $n$ over an alphabet $\Sigma$ of size $\sigma>1$ with the letters of $x$ being independent and identically distributed random variables, uniformly distributed over $\Sigma$. In the context of molecular biology we typically have $\Sigma=\{\texttt{A,C,G,T}\}$ and so we assume $\sigma=\cO(1)$.

\paragraph{Related Works.} The problem of computing longest common prefixes with $k$-errors was first studied by Manzini for $k=1$ in~\cite{Manzini:2015:LCP:2952649.2952678}. We distinguish the following techniques that can be applied to solve this and other related problems.
\begin{description}
\item[Non-constant $k$ and $\omega(n)$ space:] In this case, we can make use of the well-known data structure by Cole et al~\cite{Cole:2004:DMI:1007352.1007374}. The size of the data structure is $\cO(n\frac{(c\log n)^k}{k!})$, where $c>1$ is a constant.  
\item[Constant $k$ and $\cO(n)$ space:] In this case, we can make use of the technique by Thankachan et al~\cite{DBLP:journals/jcb/ThankachanAA16} which builds heavily on the data structure by Cole et al. The working space is {\em exponential in $k$} but $\cO(n)$ for $k=\cO(1)$.
\item[Non-constant $k$ and $\cO(n)$ space:] In this case, there exists a simple $\cO(n^2k)$-time worst-case algorithm to solve the problem. The best-known average-case algorithm was presented in~\cite{Sofsem2018}. It requires $\cO(n (\sigma R)^k \log \log n (\log k+ \log \log n) )$ time on average, where $R=\lceil (k+2) (\log_{\sigma} n+1) \rceil$.
\item[Other related works:] In~\cite{tania_arxiv} it was shown that a strongly subquadratic-time algorithm for the longest common substring with $k$-errors problem, for $k=\Omega(\log n)$ and binary strings, refutes the Strong Exponential Time Hypothesis. Thus subquadratic-time solutions for approximate variants of the problem have been developed~\cite{tania_arxiv}. A non-deterministic algorithm is also known~\cite{Abboud:2015:MAP:2722129.2722146}.
\end{description}

\paragraph{Our Contribution.}  In this paper, we continue the line of research for non-constant $k$ and $\cO(n)$ space to investigate the limits of computation in the average-case setting; in particular in light of the worst-case lower bound shown in~\cite{tania_arxiv}. We make the following threefold contribution.
\begin{enumerate}
\item We first show a non-trivial upper bound of independent interest: the expected length of the maximal longest common prefix with $k$-errors between a pair of suffixes of $x$ is $\cO(\log_\sigma n)$ when $k \leq \frac{\log n}{\log\log n}$.
\item By applying this result, we significantly improve upon the state-of-the-art algorithm for non-constant $k$ and using $\cO(n)$ space~\cite{Sofsem2018}. Specifically, our algorithm runs in $\cO(n \log^k n \log \log n)$ time on average using $\cO(n)$ space.
\item Notably, we extend our results to the {\em edit distance} model with no extra cost thus solving the genome mappability data structure problem, the longest common substring with $k$-errors problem, and the all-pairs suffix/prefix overlaps with $k$-errors problem in {\em strongly sub-quadratic} time for $k \leq \frac{\log n}{\log\log n}$.
\end{enumerate}

\section{Preliminaries}\label{sec:prel}

We begin with some basic definitions and notation.  Let $x=x[0]x[1]\ldots x[n-1]$ be a \textit{string} of length $|x|=n$ over a finite ordered alphabet $\Sigma$ of size $|\Sigma|=\sigma=\cO(1)$. 
For two positions $i$ and $j$ on $x$, we denote by $x[i\dd j]=x[i]\ldots x[j]$ the \textit{substring} 
(sometimes called \textit{factor}) of $x$ that 
starts at position $i$ and ends at position $j$. 
We recall that a {\em prefix} of $x$ is a substring that starts at position 0 
($x[0\dd j]$) and a {\em suffix} of $x$ is a substring that ends at position $n-1$ 
($x[i\dd n-1]$). 

  Let $y$ be a string of length $m$ with $0<m\leq n$. 
  We say that there exists an \textit{occurrence} of $y$ in $x$, or, more 
simply, that $y$ \textit{occurs in} $x$, when $y$ is a substring of $x$.
  Every occurrence of $y$ can be characterised by a starting position in $x$. 
  We thus say that $y$ occurs at the \textit{starting position} $i$ in $x$ when $y=x[i \dd i + m - 1]$.
  
  The {\em Hamming distance} between two strings $x$ and $y$, with $|x| = |y|$, is defined as $d_H(x, y) = |\{i : x[i] \neq y[i],\, i = 0, 1,\ldots, |x| - 1\}|$. If $|x| \neq |y|$, we set $d_H(x, y)=\infty$. The \emph{edit distance} between $x$ and $y$ is the minimum total cost of a sequence of edit operations (insertions, deletions, substitutions) required to transform $x$ into $y$. It is known as \emph{Levenshtein distance} for unit cost operations. 
We consider this special case here. If two strings $x$ and $y$ are at (Hamming or edit) distance at most $k$ we say that $x$ and $y$ have {\em $k$-errors} or have {\em at most $k$ errors}.

We denote by \SA{} the {\em suffix array} of $x$. \SA{} is an integer array of size $n$ storing the starting positions of all (lexicographically) sorted non-empty suffixes of $x$, i.e.~for all 
$1 \leq  r < n$ we have $x[\SA{}[r-1] \dd n-1] < x[\SA{}[r] \dd n - 1]$~\cite{SA}.
  Let \lcp{}$(r, s)$ denote the length of the longest common prefix between
$x[\SA{}[r] \dd n - 1]$ and $x[\SA{}[s] \dd n - 1]$ 
for positions $r$, $s$ on $x$.
  We denote by \LCP{} the {\em longest common prefix} array of $x$ defined by 
\LCP{}$[r]=\lcp{}(r-1, r)$ for all $1 \leq r < n$, and 
\LCP{}$[0] = 0$. The inverse \iSA{} of the array \SA{} is defined by 
$\iSA{}[\SA{}[r]] = r$, for all $0 \leq r < n$. It is known that
  \SA{}, \iSA{}, and \LCP{} of a string of length $n$, over a constant-sized alphabet, can be computed in time and space $\cO(n)$~\cite{Nong:2009:LSA:1545013.1545570,indLCP}.
  It is then known that a range minimum query (RMQ) data structure over the \LCP{} array, that can be constructed in $\cO(n)$ time and $\cO(n)$ space~\cite{Bender2000}, can answer \lcp{}-queries in $\cO(1)$ time per query~\cite{SA}. The \lcp{} queries are also known as {\em longest common extension} (\LCE{}) queries.

The {\em permuted \LCP{} array}, denoted by \PLCP{}, has the same contents as the \LCP{} array but in different order. Let $i^{-}$ denote the starting position of the lexicographic predecessor of $x[i\dd n-1]$. For $i = 0,\ldots,n-1$, we define 
$\PLCP[i] = \LCP[\iSA[i]] = \lcp(\iSA[i^{-}], \iSA[i]])$,
that is, $\PLCP[i]$ is the length of the longest common prefix between $x[i\dd n-1]$ and its lexicographic predecessor. For the starting position $j$ of the lexicographically smallest suffix we set $\PLCP[j]=0$. For any $k \geq 0$, we define $\lcp_k(y, z)$ as the largest  $\ell \geq 0$ such that $y[0\dd \ell - 1]$ and $z[0\dd \ell-1]$ exist and are at {\em Hamming distance} at most $k$; note that this is defined for a pair of strings. We analogously define the {\em permuted \LCP{} array with $k$-errors}, denoted by $\PLCP_k$. For $i = 0,\ldots,n-1$, we have that $$\PLCP_k[i]=\max_{j=0,\ldots,n-1,~j \neq i} \lcp_k(x[i \dd n-1], x[j\dd n-1]).$$ The main computational problem in scope can be formally stated as follows.

{\defproblem{\textsc{$\PLCP$ with $k$-Errors}}{A string $x$ of length $n$ and an integer $0<k<n$}
{$\PLCP_k$ and $\POS_k$; $\POS_k[i]\neq i$, for $i = 0,\ldots,n-1$, is such that $x[i \dd i+\ell-1] \approx_k x[\POS_k[i] \dd \POS_k[i]+\ell-1]$, where $\ell=\PLCP_k[i]$}}

We assume that $k \leq \frac{\log n}{\log \log n}$ throughout, since all relevant time-complexities contain an $n\log ^k n$ factor and any larger $k$ would force this value to be $\Omega(n^2)$: $$n\log^k n \leq c n^2 \Leftrightarrow k\log\log n \leq \log (c n) \Leftrightarrow k \leq \frac{\log c+\log n}{\log\log n},\; c\geq 1.$$

\section{Computing $\PLCP_k$}~\label{sec:hamming}

In this section we propose a new algorithm for the \textsc{$\PLCP$ with $k$-Errors} problem under both the Hamming and the edit distance (Levenshtein distance) models. This algorithm is based on a deeper look into the expected behaviour of the longest common prefixes with $k$-errors. This in turn allows us to make use of the $y$-fast trie, an efficient data structure for maintaining integers from a bounded domain. We already know the following result for errors under the Hamming distance model.

\begin{theorem}[\cite{Sofsem2018}]\label{the:PLCP_old}
Problem \textsc{$\PLCP$ with $k$-Errors} for $1 \leq k \leq \frac{\log n}{\log \log n}$ can be solved in average-case time $\cO(n (\sigma R)^k \log^2 \log n)$, where $R=\lceil (k+2) (\log_{\sigma} n+1) \rceil$, using $\cO(n)$ extra space.
\end{theorem}

In the rest of this section, we show the following result for errors under both the Hamming and the edit distance models.

\begin{restatable}{theorem}{PLCPnew}\label{the:PLCP_new}
Problem \textsc{$\PLCP$ with $k$-Errors} can be solved in average-case time $\cO(n \frac{c^k}{k!} \log^k n \log\log n)$, where $c$ is a constant, using $\cO(n)$ extra space.
\end{restatable}

For clarity of presentation, we first do the analysis and present the algorithm under the Hamming distance model in Sections~\ref{sec:exp} and~\ref{sec:yfast}. We then show how to extend our technique to work under the edit distance model in Section~\ref{sec:edit}.

\subsection{Expectations}\label{sec:exp}

The expected maximal value in the $\LCP{}$ array is $2\log_{\sigma} n + \cO(1)$~\cite{KGOTK83}.
We can thus obtain a trivial $\cO(k \log_{\sigma} n)$ bound on the expected length of the maximal longest common prefix with $k$-errors for arbitrary $k$ and $\sigma$. By looking deeper into the expected behaviour of the longest common prefixes with $k$-errors we show the following result of independent interest for when $k \leq \frac{\log n}{\log \log n}$.

\begin{theorem}\label{cyril2}
Let $x$ be a string of length $n$ over an alphabet of size $\sigma>1$
and $1 \leq k \leq \frac{\log n}{\log \log n}$ be an integer.

\begin{itemize}
\item[(a)]The expected length of the maximal longest common prefix with $k$-errors between a pair of suffixes of $x$ is $\cO(\log_{\sigma} n)$.

\item[(b)]There exists a constant $\alpha$ such that the expected number of pairs of suffixes of $x$ with a common prefix with $k$-errors of length at least $\alpha \log_\sigma n$ is $\cO(1)$.
\end{itemize}
\end{theorem}
\begin{proof}[a]
Let us denote the $i$th suffix of $x$ by $x_i=x[i \dd n-1]$. 
Further let us define the following random variables: 

$$X_{i,j}={\lcp_k(x_i,x_j)} \text{ and } Y=\max\limits_{0\leq i<j\leq n-1}{X_{i,j}}.$$ 

\begin{claim}
$\Pr(X_{i,j}\geq m) \leq \genfrac(){0pt}{0}{m}{k}\frac{1}{\sigma^{m-k}}$.
\end{claim}
\begin{proof}[of Claim]
Each possible set of positions where a substitution is allowed is a subset of one of the $\genfrac(){0pt}{1}{m}{k}$ subsets of $m$ of size $k$.
For each of these subsets, we can disregard what happens in the $k$ chosen positions; in order to yield a match with $k$-errors, the remaining $m-k$ positions must match and each of them matches with probability $\frac{1}{\sigma}$. The claim follows by applying the Union-Bound (Boole's inequality).\qed
\end{proof}

\noindent By applying the Union-Bound again we have that
$$\Pr(Y\geq m)=\Pr(\bigcup\limits_{i < j}\{X_{i,j} \geq m\})\leq \sum\limits_{i < j}{\Pr(X_{i,j}\geq m)}\leq n^2 \genfrac(){0pt}{0}{m}{k}\frac{1}{\sigma^{m-k}},$$
for $m \geq k$ and $\Pr(Y\geq m)=1$ for $m \leq k$.
The expected value of $Y$ is given by:
$$E[Y]=\sum\limits_{m = 1}^{\infty}{\Pr(Y\geq m)}=
\underbrace{\sum\limits_{m = 1}^{\alpha (\log_{\sigma}+k)}{\Pr(Y\geq m)}}_{\leq \alpha (\log_{\sigma}n+k)}
+
\sum\limits_{m = \alpha (\log_{\sigma}n+k)+1}^{\infty}{\Pr(Y\geq m)}.$$

\noindent (Note that we bound the first summand using that $\Pr(Y\geq m) \leq 1$ for all $m$.)

\begin{claim}Let $r_{m,k}=\genfrac(){0pt}{0}{m}{k}$. We have that  $\frac{r_{m,k}}{r_{m-1,k}} \leq \frac{3}{2}$ for $m \geq 6k$. 
\end{claim}
\begin{proof}[of Claim] 
$$\frac{r_{m,k}}{r_{m-1,k}}=\frac{m}{(m-k-1)}\leq\frac{6k}{(6k-k-1)}=\frac{6k}{(5k-1)}\leq \frac{3}{2}.$$\qed
\end{proof}

By assuming $\beta=\alpha (\log_{\sigma}n+k)+1 \geq 6k$, for some $\alpha > 1$, we apply the above claim to bound the second summand as follows. 
$$\sum\limits_{m = \beta}^{\infty}{\Pr(Y\geq m)}\leq \sum\limits_{m = \beta}^{\infty}{n^2 \genfrac(){0pt}{0}{m}{k}\frac{1}{\sigma^{m-k}}}
\leq
\sum\limits_{m = \beta}^{\infty}{n^2 \genfrac(){0pt}{0}{6k}{k} \left(\frac{3}{2}\right)^{m-6k} \frac{1}{\sigma^{m-k}}}$$
$$=\sum\limits_{m = \beta}^{\infty}{n^2 \genfrac(){0pt}{0}{6k}{k} \left(\frac{2}{3}\right)^{5k} \left(\frac{3}{2\sigma}\right)^{m-k}}
\leq n^2 \genfrac(){0pt}{0}{6k}{k} \left(\frac{2}{3}\right)^{5k}\left(\frac{3}{2\sigma}\right)^{\beta-k}\sum\limits_{m = 0}^{\infty}{\left(\frac{3}{2\sigma}\right)^{m}}$$
$$\leq An^6\left(\frac{3}{2\sigma}\right)^{\beta-k}
\leq An^6\left(\frac{3}{2\sigma}\right)^{\alpha \log_{\sigma}n}
= \frac{An^6}{n^{\alpha(1-1/\log_{3/2}\sigma)}}$$
for some constant $A$ since $\sigma \geq 2$ and $(6k)^k \leq 2^{3 \log n} k^k \leq n^3 \log^k n= \cO(n^4)$. Then $1-1/\log_{3/2}\sigma>0$ and we can thus pick an $\alpha$ large enough such that this sum is $\cO(n^{-\epsilon})$ for any $\epsilon>0$. 
\qed 
\end{proof}
\begin{proof}[b] Let $I_{i,j,m}$ be the indicator random variable for the event $\{X_{i,j}\geq m\}$. We then have that
$$E[\sum\limits_{i < j}
{I_{i,j,m}}]=\sum\limits_{i < j}{E[I_{i,j,m}]}=\sum\limits_{i<j}
{\Pr(X_{i,j}\geq m)},$$
which we have already shown is $\cO(1)$ if $m=\alpha \log_\sigma n$ for some $\alpha>1$.\qed
\end{proof}

\subsection{Improved Algorithm for Hamming Distance}\label{sec:yfast}

The $y$-fast trie, introduced in~\cite{DBLP:journals/ipl/Willard83}, supports insert, delete and search (exact, predecessor and successor queries) in time $\cO(\log \log U)$ {\em with high probability}, using $\cO(n)$ space, where $n$ is the number of stored values and $U$ is size of the universe. We consider each substring of $x$ of length at most $\lambda=\alpha \log n$ for a constant $\alpha$ satisfying Theorem~\ref{cyril2} (b) as a $\lambda$-digit number; note that by our assumptions this number fits in a computer word.
We thus have $U=\sigma^{\lambda}$ and hence $\log\log U = \cO(\log \log n +\log \alpha) = \cO(\log \log n)$.

We initialise $\PLCP_k$ and $\POS_k$ for each $i$ based on the longest common prefix of $x[i \dd n-1]$ (i.e.~not allowing any errors) that occurs elsewhere using the \SA{} and the \LCP{} array; this can be done in $\cO(n)$ time. For each pair of suffixes that share a prefix of at least $\lambda$ we perform (at most) $k$ \LCE{} queries to find their longest common prefix allowing for $k$-errors; by Theorem~\ref{cyril2} these pairs are $\cO(1)$.

We then initialise the $y$-fast trie by inserting $x[i \dd i + \lambda-1]$ to it for each position $i$ of $x$ with $i \leq n- \lambda$. (For the rest of the positions, for which we reach the end of $x$, we insert $x[i \dd n-1]$ after some trivial technical considerations.) This procedure takes time $\cO(n\log \log n)$ in total.

We then want to find a longest prefix of the $\sigma^k \genfrac(){0pt}{1}{\lambda}{k}$ strings of length at most $\lambda$ that are at Hamming distance at most $k$ from $x[i \dd i+\lambda-1]$ that occurs elsewhere in $x$ as well as an occurrence of it. If this prefix is of length $\lambda$, we find all positions $t$ in $x$ for which $d_H(x[i \dd i+\lambda-1],x[t \dd t+\lambda-1]) \leq k$ and treat each of them individually.
We generate a subset of the $\sigma^k \genfrac(){0pt}{1}{\lambda}{k}$ strings; we avoid generating some that we already know do not occur in $x$. We only want to allow the first error at position $p$, where $i \leq p \leq i+\PLCP_0[i]$. Let us denote the substitution at position $j$ with letter $a$ by $(j,a)$. Suppose that the longest prefix of $x[i \dd n-1]$ after substitutions $(j_0,a_0),\ldots,(j_e,a_e)$ that occurs elsewhere in the string is of length $m$. We then want to allow the $(e+1)$th error at positions $j_e<p \leq i+m$; inspect Figure~\ref{fig:err} for an illustration. It should be clear that we obtain each possible sequence of substitutions at most once.

\begin{figure}[!t]
  \begin{center}
  \begin{tikzpicture}[xscale=0.6]
    \foreach \x/\c in {0/i,3/j_0,9.1/j_e,9/,13/i+m-1}{
      \draw (\x,0) node[above] {$\c$};
    }
    \draw[-latex] (0,0.7) -- (2.4,0.7);
    \draw[-latex,decorate,decoration={snake,amplitude=.4mm,segment length=2mm}] (3.8,0.7) -- (8.4,0.7);
    \draw[-latex] (10,0.7) -- (13,0.7);
    \draw (3.1,0.45) node[above] {$a_0$};
    \draw (9.1,0.45) node[above] {$a_e$};
    \draw (10.3,0.7) node[above] {$X$};
    \draw (11.3,0.7) node[above] {$X$};
    \draw (12.3,0.7) node[above] {$X$};
    \draw (13.3,0.7) node[above] {$X$};
    \draw[snake=brace] (10,1.2) -- node[above] {$p$} (13.6,1.2);
  \end{tikzpicture}
  \end{center}
  \caption{The $(e+1)$th error is any possible substitution at a position $j_e<p \leq i+m$.
  \label{fig:err}}
\end{figure}
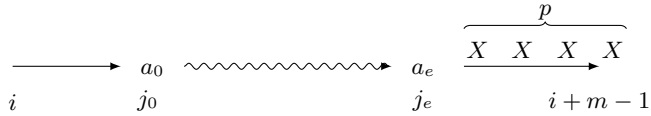

We view each string $z$ created after at most $k$ substitution operations as a number; the aim is to find its longest prefix that occurs elsewhere in $x$. To this end we perform at most three queries over the $y$-fast trie: an exact; a predecessor; and a successor query.
If the exact query is unsuccessful, then either the predecessor or the successor query will return a factor $z'$ of $x$ that attains the maximal longest common prefix that any factor of $x$ has with $z$. Note that it may be the case that $z'$ only occurs at position $i$; however in this case $\lcp_k(z,z')$ will be smaller or equal than the value currently stored at $\PLCP_k[i]$ due to how we generate each such string $z$. Hence we do not perform an invalid update of $\PLCP_k[i]$.

Having found $z'$, we can then compute the length of the longest common prefix between $z$ and $z'$ in constant time using standard bit-level operations. For clarity of presentation we assume $|z|=|z'|=\lambda$. An \textsf{XOR} operation between $z$ and $z'$ provides us with an integer $d$ specifying the positions of errors (bits set on when $d$ is viewed as binary). If $d \neq 0$, we take $\delta=\lfloor\log d\rfloor$, which provides us with the index of the leftmost bit set on which in turn specifies the length of the longest common prefix between $z$ and $z'$; specifically $\lcp_0(z,z')=\lfloor\frac{\lambda\lceil \log \sigma\rceil -\delta-1}{\lceil \log \sigma\rceil}\rfloor$. 

If $z=z'$ we perform \LCE{} queries between all suffixes of the text that have $z$ as a prefix and $x[j \dd j+\lambda-1]$; by Theorem~\ref{cyril2} we expect this to happen $\cO(1)$ times in total, so the cost is immaterial.

We have $\alpha \log n$ positions where we need to consider the $k$ errors, yielding an overall time complexity of $\cO(n \sigma^k \genfrac(){0pt}{1}{\alpha \log n}{k} \log \log n)=\cO(n \frac{(\alpha \sigma)^k}{k!} \log^k n \log\log n)$. We thus obtain the following result.

\PLCPnew*

\begin{remark}
We have that $\frac{(\alpha \sigma)^k}{k!}\leq (\alpha \sigma)^{\alpha \sigma}=\cO(1)$ and hence the required time is bounded by $\cO(n \log^k n \log\log n)$.
\end{remark}

\begin{remark}\label{rem:bille}
If $\alpha\log n > w$, where $w$ is the word size in the word-RAM model, we can make use of the {\em deterministic} data structure presented in~\cite{DBLP:conf/cpm/BilleGS17} (Theorem 1 therein), which can be built in $\cO(n)$ time for a string $x$ of length $n$ and answers predecessor queries (i.e.~given a query string $p$, it returns the lexicographically largest suffix of $x$ that is smaller than $p$) in time $\cO(\frac{|p| \log \sigma}{w}+\log |p| + \log \log \sigma)$. In particular, the queries in scope can be answered in time $\cO(\log \log n)$ per query.
\end{remark}
%
%
\subsection{Edit Distance}\label{sec:edit}

We next consider computing $\PLCP_k$ under the edit distance model; however in this case we observe that $x[i \dd n-1]$ and $x[i+j \dd n-1]$ are at edit distance $j$ for $i-k \leq j \leq i+k$. We hence alter the definition so that $\PLCP_k[i]$ refers to the longest common prefix of $x[i \dd n-1]$ with $k$-errors occurring at a position $j \notin S_{i,k}=\{i-k, \ldots, i+k\}$.

The proof of Theorem~\ref{cyril2} can be extended to allow for $k$-errors under the edit distance. In this case we have that $\Pr(X_{i,j}\geq m) \leq \genfrac(){0pt}{1}{m}{k}\frac{3^k}{\sigma^{m-k}}$; this can be seen by following the same reasoning as in the first claim of the proof with two extra considerations: (a) each deletion/insertion operation conceptually shifts the letters to be matched (giving the $3^k$ factor); (b) the letters to be matched are $m$ minus the number of deletions and substitutions and hence at least $m-k$. The extra $3^k$ factor gets consumed by $(2/3)^{5k}$ later in the proof since $2\sqrt[5]{3}/3<1$.

On the technical side, we modify the algorithm of Section~\ref{sec:yfast} as follows:
\begin{enumerate}
\item At each position, except for $\sigma-1$ substitutions, we also  consider $\sigma$ insertions and $1$ deletion. This yields a multiplicative $2^k$ factor in the time complexity. We keep counters $\textsf{ins}$ for insertions and $\textsf{del}$ for deletions; for each length we obtain, we add $\textsf{del}$ and subtract $\textsf{ins}$.
\item When querying for a string $z$ while processing position $i$ we now have to check that we do not return a position $j \in S_{i,k}$.
We can resolve this by spending $\cO(k)$ time for each position $i$; when we start processing position $i$, we create an array of size $\cO(k)$ that stores for each position $j \in S_{i,k}$ a position $f_j \notin S_{i,k}$ with the maximal longest common prefix with $x[j \dd n-1]$ using the $\SA$ and the $\LCP$ array. When a query returns a position $j \in S_{i,k}$ we instead consider $f_j$.
\item We replace the \LCE{} queries used to compute values in $\PLCP_k$ longer than $\lambda$ (that required $\cO(k)$ time in total) by the Landau-Vishkin technique~\cite{LV85} to perform extensions. For an illustration inspect Figure~\ref{fig:LV}. We initiate $2k+1$ diagonal paths in the classical dynamic programming matrix for $x[i\dd n-1]$ and $x[j\dd n-1]$. The $i$th diagonal path above and the $i$th diagonal path below the main diagonal are initialised to $i$ errors. The path starting at the main diagonal is initialised to $0$ errors. We first perform an \LCE{} query between $x[i\dd n-1]$ and $x[j+d\dd n-1]$, for all $0 \leq d \leq k$, and an \LCE{} query between $x[i+d\dd n-1]$ and $x[j\dd n-1]$, for all $1 \leq d \leq k$. Then, for all $1 \leq d \leq k$, we try to extend a path with exactly $d-1$ errors to a path with exactly $d$ errors. We perform an insertion, a deletion, or a substitution with a further \LCE{} query and pick the farthest reaching extension. The bottom-most extension of any diagonal when $d = k$ specifies the length of the longest common prefix with $k$-errors. The whole process takes time $\cO(k^2)$.
\end{enumerate}

\begin{figure}[t]
\begin{center}
\begin{tikzpicture}[scale=0.7]
	\draw (1,1) rectangle (7,7);
    \draw[very thick] (1,5.5) -- (5.5,1);
    \draw[very thick] (2.5,7) -- (7,2.5);

    \foreach \x/\c in {1.2/j,2.5/j+k,7/n-1}{
      \draw (\x,7) node[above] {\footnotesize $\c$};
    }   
        \foreach \x/\c in {6.8/i,5.5/i+k,1/n-1}{
      \draw (1,\x) node[left] {\footnotesize $\c$};
    }
\draw (4,8.5) node {$x[j\dd n-1]$};
\draw (0,4) node[left] {$x[i\dd n-1]$};
\end{tikzpicture}
\end{center}
\caption{We need to perform $k$ extension steps in $2k+1$ diagonals of the dynamic programming matrix for $x[i\dd n-1]$ and $x[j\dd n-1]$.}
\label{fig:LV}
\end{figure}
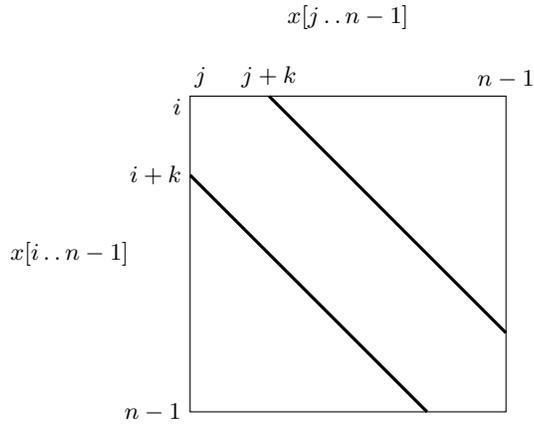

\section{Genome Mappability Data Structure}


The \emph{genome mappability} problem has already been studied under the Hamming distance model~\cite{biopaper,Manzini:2015:LCP:2952649.2952678,1map}. We can also define the problem under the edit distance model. Given a string $x$ of length $n$ and integers $m<n$ and $k<m$, we are asked to count, for each length-$m$ substring $x[i \dd i+m-1]$ of $x$, the number \textit{occ} of other substrings of $x$ occurring at a position $j \notin S_{i,k}=\{i-k, \ldots, i+k\}$ that are at edit distance at most $k$ from $x[i \dd i+m-1]$. We then say that this substring has $k$-mappability equal to $\textit{occ}$. Specifically, we consider a data structure version of this problem~\cite{Sofsem2018}. Given $x$ and $k$, construct a data structure, which, for a query value $\mu$ given on-line, returns the minimal value of $m$ that forces at least $\mu$ length-$m$ substrings of $x$ to have $k$-mappability equal to $0$. 


\begin{theorem}[\cite{Sofsem2018}]\label{thm:gmDS}
An $\cO(n)$-sized data structure answering genome mappability queries in $\cO(1)$ time per query can be constructed from $\PLCP_k$ in time $\cO(n)$.
\end{theorem}

By combining Theorem~\ref{the:PLCP_new} with Theorem~\ref{thm:gmDS} we obtain the first efficient algorithm for the genome mappability data structure under the edit distance model.

\section{Longest Common Substring with $k$-Errors}\label{sec:kLCF}

In the {\em longest common substring with $k$-errors} problem  we are asked to find the longest substrings of two strings that are at distance at most $k$. 
The Hamming distance version has received much attention due to its applications in computational biology~\cite{DBLP:journals/jcb/UlitskyBTC06,kmacs,ALFRED}.
Under edit distance, the problem is largely unexplored.
The {\em average $k$-error common substring} is an alignment-free method based on this notion for measuring string dissimilarity under Hamming distance; we denote the induced distance by $\text{Dist}_k(x,y)$ for two strings $x$ and $y$ (see~\cite{DBLP:journals/jcb/UlitskyBTC06} for the definition). $\text{Dist}_k(x,y)$ can be computed in time $\cO(|x|+|y|)$ from arrays $\Lambda_{x,y}$ and $\Lambda_{y,x}$, defined as

$$\Lambda_{x,y}[i]=\max_{0 \leq j \leq |y|-1}(\lcp_k(x[i\dd |x|-1],y[j\dd|y|-1])).$$

A worst-case and a more practical average-case algorithm for the computation of $\Lambda_{x,y}$ have been presented in~\cite{DBLP:journals/jcb/ThankachanAA16,ALFRED}. This measure was extended to allow for {\em wildcards} (don't care letters) in the strings in~\cite{DBLP:journals/nar/HorwegeLBHKLM14}. Here we provide a natural generalisation of this measure: the average $k$-error common substring under the edit distance model. The sole change is in the definition of $\Lambda_{x,y}[i]$: except for substitution, we also allow for insertion and deletion operations. 


The algorithm of Section~\ref{sec:edit} can be applied to compute $\Lambda_{x,y}$ under the edit distance model within the same complexities. We start by constructing the $y$-fast trie for $y$. We then do the queries for the suffixes of $x$; we now also check for an exact match (i.e.~for $x[i \dd i+\alpha \log (|x|+|y|)-1]$). We obtain the following result.

\begin{theorem}\label{thm:lcs}
Given two strings $x$ and $y$ of length at most $n$ and a distance threshold $k$, arrays $\Lambda_{x,y}$ and $\Lambda_{y,x}$ and $\text{Dist}_k(x,y)$ can be computed in average-case time $\cO(n \frac{c^k}{k!} \log^k n \log \log n) $, where $c$ is a constant, using $\cO(n)$ extra space.
\end{theorem}

\begin{remark}
By applying Theorem~\ref{thm:lcs} we essentially solve the longest common substring with $k$-errors for $x$ and $y$ within the same complexities.
\end{remark}

\section{All-Pairs Suffix/Prefix Overlaps with $k$-Errors}

Given a set of strings and an error-rate $\epsilon$, the goal is to find, for all pairs of strings, their suffix/prefix matches (overlaps) that are within distance $k=\lceil \epsilon\ell \rceil$, where $\ell$ is the length of the overlap~\cite{DBLP:journals/jcb/RasmussenSM06,DBLP:journals/iandc/ValimakiLM12,DBLP:conf/spire/KucherovT14}. 

Using our technique but only inserting {\em prefixes} of the strings in the $y$-fast trie and querying for all starting positions ({\em suffixes}) in a similar manner as in Section~\ref{sec:exp}, we obtain the following result.

\begin{theorem}\label{thm:ove}
Given a set of strings of total length $n$ and a distance threshold $k$, the length of the maximal longest suffix/prefix overlaps of every string against all other strings within distance $k$ can be computed in average-case time $\cO(n \frac{c^k}{k!} \log^k n \log \log n) $, where $c$ is a constant, using $\cO(n)$ extra space.
\end{theorem}

\bibliographystyle{plain}
\bibliography{references}

\end{document}